\documentclass[letterpaper,11pt]{article}
\usepackage{url}
\usepackage{amssymb,amsmath,amsfonts,amsthm}
\usepackage{fullpage}
\usepackage[usenames]{color}

\DeclareMathSymbol{\qedsymb} {\mathord}{AMSa}{"04}

\newcommand{\eps}{\varepsilon}
\renewcommand{\epsilon}{\varepsilon}

\newcommand{\oct}{\quad\quad}                                   

\newcommand{\sign}{\mathrm{sign}}

\newcommand{\R}{\mathbb{R}}

\newcommand{\E}{\mathbf{E}}

\renewcommand{\Pr}{\mathbf{Pr}}

\newcommand{\EquationName}[1]{\label{eq:#1}}

\newcommand{\LemmaName}[1]{\label{lem:#1}}

\newcommand{\SectionName}[1]{\label{sec:#1}}
\newcommand{\TheoremName}[1]{\label{thm:#1}}

\newcommand{\Equation}[1]{Eq.\:\eqref{eq:#1}}

\newcommand{\Lemma}[1]{Lemma~\ref{lem:#1}}

\newcommand{\Section}[1]{Section~\ref{sec:#1}}
\newcommand{\Theorem}[1]{Theorem~\ref{thm:#1}}

\newtheorem{theorem}{Theorem}

\newtheorem{definition}[theorem]{Definition}

\newtheorem{lemma}[theorem]{Lemma}
\newtheorem{remark}[theorem]{Remark}

\newcommand{\proofbelow}{3pt}
\newcommand{\afterproof}{\hfill $\blacksquare$ \par \vspace{\proofbelow}}
\newcommand{\aftersubproof}{\hfill $\Box$ \par \vspace{\proofbelow}}
\renewenvironment{proof}{\noindent\textbf{Proof.}\,}{\afterproof}

\newenvironment{proofof}[1]{\noindent\textbf{Proof} \,(of #1).\,}{\afterproof}

\newcommand{\poly}{\mathop{{\rm poly}}}
\renewcommand{\th}{\ifmmode{^{\textrm{th}}}\else{\textsuperscript{th}\ }\fi}

\newcommand{\comment}[1]{}

\begin{document}

\author{Daniel M. Kane\footnotemark[2]\oct
Jelani Nelson\footnotemark[3]
  }

\date{}

\footnotetext[1]{Harvard University, Department of
  Mathematics. \texttt{dankane@math.harvard.edu}.}
\footnotetext[2]{MIT Computer Science and Artificial Intelligence
  Laboratory. \texttt{minilek@mit.edu}. 
}

\title{A Derandomized Sparse Johnson-Lindenstrauss Transform}

\maketitle

\begin{abstract}
Recent work of [Dasgupta-Kumar-Sarl\'{o}s, STOC 2010] gave a
sparse
Johnson-Lindenstrauss transform and left as a main open question
whether their construction could be efficiently derandomized. We
answer their question affirmatively by giving an alternative proof of
their result requiring only bounded independence hash functions.
Furthermore, the sparsity bound obtained in our proof is 
improved. Our work implies the first implementation of a
Johnson-Lindenstrauss transform in data
streams with sublinear update time.
\end{abstract}

\section{Introduction}\SectionName{intro}
The Johnson-Lindenstrauss lemma states the following.

\begin{lemma}[JL Lemma {\cite{JL84}}]\LemmaName{jl-lemma}
For any integer $d>0$, and any $0<\eps,\delta<1/2$,
there exists a probability distribution on $k\times d$ 
real matrices for $k = \Theta(\eps^{-2}\log(1/\delta))$ 
such that for any $x\in\R^d$ with $\|x\|_2 = 1$,
$$ \Pr_A[|\|Ax\|_2^2 - 1| > \eps] <
\delta . $$
\end{lemma}

Several proofs of the JL lemma exist in the literature
\cite{Achlioptas03,AV06,DG03,FM88,IM98,JL84,Matousek08}, and it is
known that the
dependence on $k$ is tight up to a constant factor
\cite{JW11} (see also \Section{jl-optimal} for another proof).
Though, these
proofs of the JL lemma give a distribution over {\em dense} matrices,
where each column has
at least a constant fraction of its entries being non-zero, and thus
na\"{i}vely performing the matrix-vector multiplication is
costly. Recently, Dasgupta, Kumar, and Sarl\'{o}s
\cite{DKS10} proved the JL lemma where each matrix in the support of
their distribution
only has $\alpha$ non-zero entries per column, for $\alpha =
\Theta(\eps^{-1}\log(1/\delta)\log^2(k/\delta))$. This reduces the
time to perform dimensionality reduction from the na\"{i}ve $O(k\cdot
\|x\|_0)$ to $O(\alpha\cdot\|x\|_0)$, where $x$ has $\|x\|_0$ non-zero
entries.

The construction of \cite{DKS10}
involved picking two random hash functions
$h:[d\alpha]\rightarrow[k]$ and
$\sigma:[d\alpha]\rightarrow\{-1,1\}$ (we use $[n]$ to denote
$\{1,\ldots,n\}$), and thus required
$\Omega(d\alpha\cdot \log k)$
bits of seed to represent a random matrix from their JL
distribution.  
They then left two main open questions: (1) derandomize their
construction to require fewer random bits to
select a random JL matrix, for applications in e.g. streaming
settings where storing a long random seed is prohibited, and
(2) understand the dependence on $\delta$ that is required in
$\alpha$.

We give an alternative proof of the main result of \cite{DKS10} that
yields progress for both (1) and (2) above
simultaneously. Specifically, our proof yields a value of $\alpha$
that is improved by a $\log(k/\delta)$ factor.  Furthermore, our proof
only requires that $h$ be $r_h$-wise independent and $\sigma$ be
$r_\sigma$-wise independent for $r_h = O(\log(k/\delta))$ and $r_\sigma =
O(\log(1/\delta))$, and thus a random sparse JL matrix can be
represented using only $O(\log(k/\delta)\log(d\alpha + k)) =
O(\log(k/\delta)\log d)$ bits (note $k$ can be
assumed less than $d$, else the JL lemma is trivial, in which case
also $\log(d\alpha) = O(\log d)$). 
We remark that \cite{DKS10} asked exactly this question: whether the
random hash functions used in their construction could be replaced by
functions from bounded independence hash families. 
The proof in \cite{DKS10} required use of the FKG inequality
\cite[Theorem 6.2.1]{AlonSpencer00}, and they suggested
that one approach to a proof that bounded independence suffices might
be to prove some form of this
inequality under bounded independence.  Our approach is completely
different, and does not use the FKG inequality at all.
Rather, the main
ingredient in our proof is the Hanson-Wright inequality \cite{HW71},
a central moment bound for quadratic forms in terms of the Frobenius
and operator norms of the associated matrix.

We now give a formal statement of the main theorem of this work,
which is a derandomized JL lemma where every matrix in the support of
the distribution has good column sparsity.

\begin{theorem}[Main Theorem]\TheoremName{main}
For any integer $d>0$, and any $0<\eps,\delta<1/2$,
there exists a family $\mathcal{A}$ of $k\times d$ 
real matrices for $k = \Theta(\eps^{-2}\log(1/\delta))$ 
such that for any $x\in\R^d$,
$$ \Pr_{A\in\mathcal{A}}[\|Ax\|_2 \notin [(1-\eps)\|x\|_2,(1+\eps)\|x\|_2]] <
\delta . $$
and where $A\in\mathcal{A}$ can be sampled using $O(\log(k/\delta)\log
d)$ random bits.
Every matrix in $\mathcal{A}$ has at most $\alpha =
\Theta(\eps^{-1}\log(1/\delta)\log(k/\delta))$ non-zero entries
per column, and thus $Ax$ can be evaluated in $O(\alpha\cdot \|x\|_0)$
time if $A$ is written explicitly in memory. If $A\in\mathcal{A}$ is
not written explicitly in memory but rather
we are given a string of $\log(|\mathcal{A}|)$ representing some matrix $A\in
\mathcal{A}$, then the multiplication
$Ax$ can be performed in $O(\alpha\cdot \|x\|_0 + t(\alpha\cdot
  \|x\|_0, O(\log(k/\delta)),
  d\alpha,k) + t(\alpha\cdot \|x\|_0, O(\log(1/\delta)), d\alpha,
  2)$ time.  Here $t(s, r, n,m)$ is the total time required to
evaluate a random hash function drawn from an $r$-wise independent
family mapping $[n]$ into $[m]$ on $s$ inputs.
\end{theorem}

We stated the time to multiply $Ax$ above in terms of the
$t(\cdot)$ function since one can evaluate an $r$-wise independent
hash function on multiple points quickly via polynomial fast
multipoint
evaluation.  Specifically, an $r$-wise independent hash family over
a finite field can consist of degree-$(r-1)$ polynomials, and a
degree-$(r-1)$ polynomial over a field
can be evaluated on $r-1$ points in only $O(r\log^2
r\log\log r)$ field operations as opposed to $O(r^2)$
operations \cite[Ch. 10]{GG99}.

We also show a variant of our main result: that it is also
possible to take $\alpha =
\eps^{-(1+o_{\delta}(1))}\log^2(1/\delta)$ and set $r_h = r_\sigma =
O(\log(1/\delta))$.  Here $o_{\delta}(1)$ denotes a function that goes
to $0$ as $\delta\rightarrow 0$ (specifically the function is
$O(1/\log(1/\delta))$. This matches the best previously known
seed length for JL of $O(\log(1/\delta)\log d)$ bits, and we still
achieve good column sparsity.

\paragraph{Implication for the streaming model.} In the turnstile
model of streaming \cite{Muthu}, a
high-dimensional vector $x\in\R^d$ receives several updates of the
form ``$(i,v)$'' in a stream which causes the change $x_i\leftarrow
x_i + v$, where $(i,v)\in\{1,\ldots,n\}\times
\{-M,\ldots,M\}$ for some positive integer $M$.

Indyk first considered the problem of maintaining a low-dimensional
$\ell_2$-embedding of $x$ in the turnstile model in \cite{Indyk06},
where he
suggested using a pseudorandom Gaussian matrix generated using Nisan's
pseudorandom generator (PRG) \cite{Nisan92}.\footnote{As noted in
  \cite{Indyk06}, the AMS sketch of \cite{AMS99} does not give an
  $\ell_2$ embedding since the median operator is used to achieve low
  error probability.}  Clarkson and Woodruff
 later showed that the entries can be $r$-wise
independent Bernoulli for $r = O(\log(1/\delta))$ \cite{CW09}. Both of
these
results though give an algorithm whose {\em update time} (the time
required to process a stream update) is $\Omega(k)$.  Using
the matrix  of \cite{DKS10} would give an update time
of $O(\alpha)$ (in addition to the time required to 
evaluate the $O(\log(k/\delta))$-wise independent hash function),
except that their construction requires superlinear
($\Omega(d\alpha\log k)$) space to store the hash function. As noted
in \cite{DKS10}, it is unclear how to use Nisan's PRG to usefully
derandomize their construction since evaluating the PRG would require
$\Omega(k)$ time.\footnote{The evaluation time is at least linear in
  the seed length, which is at least the space usage of the
  machine being fooled ($\Omega(k)$ space in this case).} Our
derandomization thus gives the first update time for $\ell_2$
embedding in data streams which is subquadratic in $1/\eps$.

\section{Related Work}\SectionName{related-work}
There have been two separate lines of related work: one line of work
on constructing JL families\footnote{In many
  known proofs of the JL lemma, the distribution over matrices in
  \Lemma{jl-lemma} is obtained by picking a matrix uniformly at random
from some set $\mathcal{A}$.  In such a case, we call $\mathcal{A}$ a
{\em JL family}.} such that the
dimensionality
reduction can be performed quickly, and another line of work on
derandomizing the JL lemma so that a random matrix from some JL family
can be selected using few random bits.  We discuss both here.

\subsection{Works on efficient JL embeddings}
Here and throughout, for a JL family $\mathcal{A}$ we use the term
{\em embedding time} to refer to the running time required to perform
a matrix-vector multiplication for an arbitrary $A\in\mathcal{A}$.
The first work to give a JL family with embedding time potentially
better than $O(kd)$ was by Ailon and Chazelle \cite{AC06}. The authors
achieved embedding time $O(d\log d + k\log^2(1/\delta))$.
Later, improvements were given by Ailon and
Liberty in \cite{AL09,AL11}. The work of \cite{AL09} achieves
embedding time $O(d\log k)$ when $k = O(d^{1/2 - \gamma})$ for an
arbitrarily small constant $\gamma > 0$, and \cite{AL11} achieves
embedding time $O(d\log d)$ and no restriction on $k$,
though the $k$ in their JL family is $O(\eps^{-4}\log(1/\delta)\log^4
d)$ as opposed to the $O(\eps^{-2}\log(1/\delta))$ bound of the
standard JL lemma. This dependence on $1/\eps$ was recently improved to
quadratic by Krahmer and Ward \cite{KW10}, though the $\log^4 d$
factor remains.
The works of Hinrichs and Vyb\'{i}ral \cite{HV10} and later Vyb\'{i}ral
\cite{Vybiral10} considered taking a random partial circulant matrix as the
embedding matrix.  This gives embedding time $O(d\log d)$ via the
Fast Fourier transform, and it was shown that one can take either $k =
O(\eps^{-2}\log^3(1/\delta))$ \cite{HV10} or $k =
O(\eps^{-2}\log(1/\delta) \log(d/\delta))$ \cite{Vybiral10}.
Liberty, Ailon, and Singer \cite{LAS08} achieve embedding time $O(d)$
when $k = O(d^{1/2 - \gamma})$, but their JL family only applies for $x$
satisfying $\|x\|_\infty \le \|x\|_2\cdot k^{-1/2}d^{-\gamma}$. 

None
of the above works however can take advantage of the situation when
$x$ is
sparse to achieve faster embedding time.  The first work which could
take advantage of sparse $x$ was that of Dasgupta, Kumar, and
Sarl\'{o}s \cite{DKS10} who
gave a JL family whose matrices all had 
$O(\eps^{-1}\log(1/\delta)\log^2(k/\delta))$ non-zero entries per
column.  They also
showed that for a large class of constructions, sparsity
$\min\{\eps^{-2}, \eps^{-1}\sqrt{\log_k(1/\delta)}\}$ is necessary
when $\delta = o(1)/d^2$. 

Other related works include \cite{CCF02} and \cite{ThorupZhang04}.
Implicitly in \cite{CCF02}, and later more explicitly in
\cite{ThorupZhang04}, a JL family was given with column sparsity $1$
using only constant-wise independent hash functions. The construction
was in fact the same as in \cite{DKS10}, but with $h$ being pairwise
independent, and 
$\sigma$ being $4$-wise independent. This construction only gives a JL
family for constant $\delta$ though, 
since with such mild independence
assumptions on $h,\sigma$ one needs $k$ to be polynomially large in
$1/\delta$.

\subsection{Works on derandomizing the JL lemma}
The $\ell_2$-streaming algorithm of Alon, Matias, and Szegedy
\cite{AMS99} implies a JL family with seed length
$O(\log d)$ and with $k = O(1/(\eps^2\delta))$.
Karnin, Rabani, and Shpilka \cite{KRS09} recently gave a family
with seed length $(1+o(1))\log_2 d + O(\log^2(1/\eps))$ also with
$k = \poly(1/(\eps\delta))$.
The best known seed length for a JL family we are aware of is due to
Clarkson and Woodruff \cite{CW09}.  Theorem 2.2 of \cite{CW09}
implies that a scaled random Bernoulli matrix with
$\Omega(\log(1/\delta))$-wise independent entries satisfies the JL
lemma, giving seed length $O(\log(1/\delta)\cdot \log d)$.
In \Section{derandomize}, we show how to bootstrap the $r$-wise
independent JL family construction to achieve seed length $O(\log d +
\log(1/\eps)\log(1/\delta) + \log(1/\delta)\log\log(1/\delta))$. We
note that a construction which achieves this seed length for
$\delta \le d^{-\Omega(1)}$ was recently achieved independently by
Meka \cite{Meka10}.

Derandomizing the JL lemma is also connected to pseudorandom
generators (PRGs) against degree-$2$
polynomial threshold functions (PTFs) over the hypercube \cite{DKN10,
  MZ10}. A degree-$t$
PTF is a function
$f:\{-1,1\}^d\rightarrow\{-1,1\}$ which can be represented as the sign of a
degree-$t$ $d$-variate polynomial.  A PRG that $\delta$-fools
degree-$t$ PTFs is a
function $F:\{-1,1\}^s\rightarrow\{-1,1\}^d$ such that for any
degree-$t$ PTF $f$,
$$ |\E_{z\in\mathcal{U}^s}[f(F(z))] - \E_{x\in\mathcal{U}^d}[f(x)]| <
\delta ,$$
where $\mathcal{U}^m$ is the uniform distribution on $\{-1,1\}^m$.

Note that the conclusion of the JL
lemma can be rewritten as
$$ \E_A[I_{[1-\eps,1+\eps]}(\|Ax\|_2^2)] \ge 1 - \delta ,$$
where $I_{[a,b]}$ is the indicator function of the interval
$[a,b]$, and furthermore $A$ can be taken to have random $\pm 1/\sqrt{k}$
entries \cite{Achlioptas03}. Noting that $I_{[a,b]}(z) = (\sign(z - a) - \sign(z -
b))/2$ and using linearity of expectation, we see that any PRG
which $\delta$-fools $\sign(p(x))$ for degree-$t$ polynomials $p$ must
also $\delta$-fool $I_{[a,b]}(p(x))$.
Now, for fixed $x$, $\|Ax\|_2^2$ is a degree-$2$ polynomial over the
boolean hypercube
in the variables $A_{i,j}$ and thus a PRG which $\delta$-fools
degree-$2$ PTFs also gives a JL family with the same seed length. Each
of \cite{DKN10,MZ10} thus give JL families with seed length
$\poly(1/\delta)\cdot \log d$. 
Also, it can be shown via the probabilistic
method that there exist PRGs for degree-$2$ PTFs with seed length
$O(\log(1/\delta) + \log d)$ (see Section B of the full version of
\cite{MZ10} for a proof), and it remains an interesting open problem
to achieve this seed length with an explicit construction. It is also
not too hard to show that any JL family $\mathcal{F}$ must have seed length
$\Omega(\log(1/\delta) + \log (d/k))$.\footnote{We need $|\mathcal{F}| \ge
  1/\delta$ to have error probability $\delta$.  Also, if
  $|\mathcal{F}| < d/k$, then the matrix obtained by
  concatenating all rows of matrices in $\mathcal{F}$ has a
  non-trivial kernel, implying a vector
exists in the intersection of all their kernels.}

Other derandomizations of the JL lemma include the works \cite{EIO02}
and \cite{Sivakumar02}. A common application of the JL lemma is the
case where there are $n$ vectors $x_1,\ldots,x_n\in\R^d$ and one wants
to find a matrix $A\in\R^{k\times d}$ to preserve $\|x_i - x_j\|_2$ to
within relative error $\eps$ for all $i,j$. In this case, one can set
$\delta = 1/n^2$ and apply \Lemma{jl-lemma}, then perform a union bound
over all $i,j$ pairs.  The works of \cite{EIO02,Sivakumar02} do not
give JL families, but rather give deterministic algorithms for finding
such a matrix $A$ in the case that the vectors $x_1,\ldots,x_n$ are
known up front.

\section{Conventions and Notation}\SectionName{notation}
\begin{definition}
For $A\in\R^{n\times n}$,
we define the {\em Frobenius norm} of $A$ as 
$\|A\|_F =
\sqrt{\sum_{i,j} A_{i,j}^2}$.
\end{definition}
\begin{definition}
For $A\in\R^{n\times n}$, we define the {\em operator norm} of $A$ as
$$\|A\|_2 = \sup_{\|x\|_2 = 1} \|Ax\|_2 .$$
In the case $A$ has all real
eigenvalues (e.g. it is symmetric), we also have that $\|A\|_2$ is the
largest magnitude of an eigenvalue of $A$.
\end{definition}

Throughout this paper, $\eps$ is the quantity given in
\Lemma{jl-lemma}, and is assumed to be smaller than some absolute
constant $\eps_0>0$. All logarithms are base-$2$ unless explicitly
stated otherwise. Also, for a positive integer $n$ we use $[n]$ to
denote the set $\{1,\ldots,n\}$.  All vectors are assumed to be column
vectors, and $v^T$ for a vector $v$ denotes its transpose.
Finally, we often implicitly assume that various quantities are powers
of $2$
(such as e.g. $1/\delta$), which is without loss of generality.

\section{Warmup: A simple proof of the JL lemma}\SectionName{warmup}
Before proving our main theorem, as a warmup we demonstrate how a
simpler version of our approach reproves Achlioptas' result \cite{Achlioptas03}
that the family of all (appropriately scaled) sign matrices is a JL
family. Furthermore, as was already demonstrated in
\cite[Theorem 2.2]{CW09}, we show that rather than choosing a
uniformly random sign matrix, the entries
need only be $\Omega(\log (1/\delta))$-wise independent.

We first state the Hanson-Wright inequality \cite{HW71},
which gives a central
moment bound for quadratic forms in terms of both the Frobenius and
operator norms of the associated matrix
\footnote{\cite{HW71} proves a tail bound, but it is
  not hard to then derive a moment bound via integration;
  see \cite{DKN10} for a direct proof of the moment bound}.

\begin{lemma}[Hanson-Wright inequality {\cite{HW71}}]\LemmaName{dkn}
Let $z = (z_1,\ldots,z_n)$ be a vector of i.i.d. Bernoulli $\pm 1$
random variables.  Then for any symmetric $B\in\R^{n\times n}$ and
integer $\ell\ge 2$ a power of $2$,
$$ \E\left[\left(z^TBz - \mathrm{trace}(B)\right)^\ell\right] \le
64^\ell\cdot
\max\left\{\sqrt{\ell}\cdot \|B\|_F, \ell\cdot \|B\|_2\right\}^\ell
.$$
\end{lemma}

\begin{theorem}\TheoremName{warmup}
For $d>0$ an integer and any $0<\eps,\delta< 1/2$, let $A$ be a $k\times d$ random
matrix with $\pm 1/\sqrt{k}$ entries that are $r$-wise independent for
$k = \Omega(\eps^{-2}\log(1/\delta))$ and $r =
\Omega(\log(1/\delta))$.  Then for any $x\in\R^d$ with $\|x\|_2 =
1$,
$$ \Pr_A[|\|Ax\|_2^2 - 1| > \eps] <
\delta . $$
\end{theorem}
\begin{proof}
We have
\begin{equation}
\|Ax\|_2^2 =
\frac 1k \cdot\sum_{i=1}^k\left(\sum_{(s,t)\in [d]\times[d]}
  x_sx_t\sigma_{i,s}\sigma_{i,t}\right)
,\EquationName{warmup-form}
\end{equation}
where $\sigma$ is a $kd$-dimensional vector formed by concatenating
the rows of $\sqrt{k}\cdot A$. Define the matrix $T\in\R^{kd\times
  kd}$ to be the block-diagonal matrix where each block equals $xx^T/k$.
Then, $\|Ax\|_2^2 = \sigma^T T \sigma$. Furthermore,
$\mathrm{trace}(T) = \|x\|_2^2 = 1$. Thus, we would like to argue that
$\sigma^T T\sigma$ is concentrated about $\mathrm{trace}(T)$,
for which we can use \Lemma{dkn}. Specifically, if $\ell\ge 2$ is even,
$$\Pr[|\|Ax\|_2^2 - 1| > \eps] = \Pr[|\sigma^T T\sigma -
\mathrm{trace}(T)| > \eps] < \eps^{-\ell}\cdot
\E[(\sigma^T T\sigma -\mathrm{trace}(T))^\ell] $$
by Markov's inequality.
To apply \Lemma{dkn},
we also pick $\ell$ a power of $2$, and we ensure $2\ell \le r$
so that the $\ell$th moment of $\sigma^T T\sigma - \mathrm{trace}(T)$ is
determined by $r$-wise independence of the $\sigma$ entries.
We also must bound $\|T\|_F$ and $\|T\|_2$. Direct
computation gives
$\|T\|_F^2 = (1/k) \cdot \|x\|_2^4 = 1/k$. Also,
$x$ is the only eigenvector of $xx^T/k$ with non-zero eigenvalue,
and furthermore its eigenvalue is $\|x\|_2^2/k = 1/k$, and thus
$\|T\|_2 = 1/k$.  Therefore,
\begin{equation}\EquationName{warmup-eq}
\Pr[|\|Ax\|_2^2 - 1| > \eps] < 64^\ell \cdot
\max\left\{\eps^{-1}\sqrt{\frac{\ell}{k}},
  \eps^{-1}\frac{\ell}{k}\right\}^\ell ,
\end{equation}
which is at most $\delta$ for $\ell = \log(1/\delta)$ and $k \ge
4\cdot 64^2\cdot \eps^{-2}\log(1/\delta)$.\footnote{Though our
  constant factor for $k$ is quite large, most likely the $64$ could
  be made much smaller by tightening the analysis of constants in
  \cite{DKN10}.}
\end{proof}

\begin{remark}
\textup{
The conclusion of \Lemma{dkn} holds even if the $z_i$ are not necessarily Bernoulli but rather have mean $0$, variance $1$, and sub-Gaussian tails, albeit with the ``64'' possibly replaced by a different constant (see \cite{HW71}). Thus, the above proof of \Theorem{warmup} carries over unchanged to show that $A$ could instead have $\Omega(\log(1/\delta))$-wise independent such $z_i$ as entries.  We also direct the reader to an older proof of this fact by Matousek \cite{Matousek08}, without discussion of independence requirements (though independence requirements can most likely be calculated from his proof, by converting the tail bounds he uses into moment bounds via integration).
}
\end{remark}

\section{Proof of Main Theorem}\SectionName{body}
We recall the sparse JL transform construction of \cite{DKS10} (though
the
settings of some of our constants differ).  Let $k =
28\cdot 64^2\cdot \eps^{-2}\log(1/\delta)$. Pick
random hash functions
$h:[d]\rightarrow[k]$ and $\sigma:[d]\rightarrow\{-1,1\}$.
Let $\delta_{i,j}$ be the indicator random variable for the event
$h(j) = i$. Define the matrix
$A\in\{-1,0,1\}^{k\times d}$ by $A_{i,j} = \delta_{i,j}\cdot
\sigma(j)$.
The work of \cite{DKS10} showed that as long as $x\in\R^d$ satisfies
$\|x\|_2 = 1$ and has bounded $\|x\|_\infty$, then
$\Pr_{h,\sigma}[|\|Ax\|_2^2 - 1| > \eps] < O(\delta)$.  We show
the same conclusion without the assumption that $h,\sigma$ are perfectly
random; in particular, we show that $h$ need only be $r_h$-wise
independent and $\sigma$ need only be $r_\sigma$-wise independent for
$r_h = O(\log(k/\delta))$ and $r_\sigma =
O(\log(1/\delta))$. Furthermore, our assumption on the bound for
$\|x\|_\infty$ is $\|x\|_\infty \le c$ for $c =
\Theta(\sqrt{\eps/(\log(1/\delta)\cdot \log(k/\delta))})$,
whereas
\cite{DKS10} required
$c = \Theta(\sqrt{\eps/(\log(1/\delta)\cdot \log^2(k/\delta))})$.
This is relevant since the column sparsity obtained in the final JL
transform construction of \cite{DKS10} is $1/c^2$. This is because, to
apply the dimensionality reduction of \cite{DKS10} to an arbitrary $x$
of unit $\ell_2$ norm (which might have $\|x\|_\infty \gg c$), one should
first map $x$ to a vector $\tilde{x}$ by a $(d/c^2)\times d$ matrix
$Q$ with $Q_{i_1r+i_2,i_1+1} = c$ and other entries $0$ for
$i_1\in\{0,\ldots,d-1\}$, $i_2\in[1/c^2]$.  Then $\|\tilde{x}\|_2 = 1$ and
$\|\tilde{x}\|_\infty \le c$, and thus the set of products with $Q$ of
JL matrices in the distribution of
\cite{DKS10} over dimension $d/c^2$ serves as a JL family for
arbitrary unit vectors.
 Thus, the
sparsity obtained by our proof in the final JL construction is
improved by a
$\Theta(\log(k/\delta))$ factor.

Before proving our main theorem, first we note that

\begin{equation*}
\|Ax\|_2^2 = \|x\|_2^2 +
2\sum_{(s,t)\in\binom{[d]}{2}}\left(\sum_{j = 1}^k
  \delta_{s,j}\delta_{t,j}x_sx_t\right)
\sigma(s)\sigma(t) .
\end{equation*}

We would like that $\|Ax\|_2^2$ is concentrated about $1$, or rather,
that
\begin{equation}\EquationName{Zdef}
Z = 2\sum_{s<t}\left(\sum_{j = 1}^k
  \delta_{s,j}\delta_{t,j}x_sx_t\right)\sigma(s)\sigma(t)
\end{equation}
is concentrated about $0$.
Let $\eta_{s,t}$ be the indicator random variable for the event $s\neq
t$ and
$h(s) = h(t)$.  Then for fixed $h$, $Z$ is a quadratic form in the
$\sigma(i)$ which can be written as $\sigma^TT\sigma$ for a
$d\times d$ matrix $T$ with
$T_{s,t} = x_sx_t\eta_{s,t}$ (we here and henceforth slightly abuse
notation by sometimes using $\sigma$ to also denote the
$d$-dimensional vector whose $i$th entry is $\sigma(i)$).

Our main theorem follows by applying \Lemma{dkn} to $\sigma^TT\sigma$,
as
in the proof of \Theorem{warmup} in \Section{warmup}, to show that $Z$ is
concentrated about $\mathrm{trace}(T) = 0$. However, unlike
in \Section{warmup}, our matrix $T$ is not a fixed matrix, but rather
is {\em random}; it depends on the random choice of $h$. We handle
this issue by using the two lemmas below, which state that
both $\|T\|_F$ and $\|T\|_2$ are small with high probability over the
random choice of $h$. We then obtain our main theorem by first
conditioning on this high probability event before applying
\Lemma{dkn}. The lemmas are proven in \Section{frobenius}
and \Section{operator}.

Henceforth in this paper, we assume $\|x\|_2 = 1$,
$\|x\|_\infty \le c$, and $T$ is the matrix described above.

\begin{lemma}\LemmaName{frobenius}
$\Pr_h[\|T\|_F^2 > 7/k] < \delta$.
\end{lemma}

\begin{lemma}\LemmaName{operator}
$\Pr_h[\|T\|_2 > \eps/(128\cdot \log(1/\delta))] < \delta$.
\end{lemma}

The following theorem now implies our main theorem (\Theorem{main}).

\begin{theorem}\TheoremName{thisisit}
$$ \Pr_{h,\sigma}[|\|Ax\|_2^2 - 1| > \eps] < 3\delta .$$
\end{theorem}
\begin{proof}
Write
\begin{align*}
\|Ax\|_2^2 & = \|x\|_2^2 +
2\sum_{(s,t)\in\binom{[d]}{2}}x_sx_t\eta_{s,t}
\sigma(s)\sigma(t)\\
&{} = 1 + Z .
\end{align*}

We will show $\Pr_{h,\sigma}[|Z| > \eps] < 3\delta$.  Condition on
$h$, and let $\mathcal{E}$
be the event that $\|T\|_F^2 \le 7/k$ and $\|T\|_2 \le
\eps/\log(1/\delta)$. By applications of
\Lemma{frobenius} and \Lemma{operator} and a union bound,
$$ \Pr_{h,\sigma}[|Z| > \eps] < \Pr_{\sigma}[|Z| > \eps\mid
\mathcal{E}] + 2\delta .$$
By a
Markov bound applied to the random variable $Z^\ell$ for $\ell$ an even
integer,
$$\Pr_\sigma[|Z| > \eps\mid \mathcal{E}] < \E_\sigma[Z^\ell\mid
\mathcal{E}] / \eps^\ell .$$
Since $Z = \sigma^TT\sigma$ and $\mathrm{trace}(T) = 0$,
applying \Lemma{dkn} with $B=T$ and $2\ell \le r_\sigma$ gives
\begin{equation}\EquationName{apply-spectral}
\Pr_\sigma[|Z| > \eps\mid \mathcal{E}] < 64^\ell\cdot
\max\left\{\eps^{-1}\sqrt\frac{7\ell}{k},
  \frac{\ell}{128\cdot \log(1/\delta)}\right\}^\ell .
\end{equation}
since the $\ell$th moment is determined by $r_\sigma$-wise
independence of $\sigma$. We conclude the proof by noting that the
expression in \Equation{apply-spectral}
is at most $\delta$ for $\ell = \log(1/\delta)$.
\end{proof}

\begin{remark}
\textup{In the proof \Theorem{thisisit}, rather than condition on
$\mathcal{E}$ we can directly bound the $O(\log(1/\delta))$th moment
of $Z$ over the randomness of both $h$ and $\sigma$
simultaneously.  In this case, we use the Frobenius and operator norm
moments from \Equation{frobenius} and \Equation{operator} directly.
This gives
$$ \Pr_{h,\sigma}[|Z| > \eps] < \eps^{-\ell}\cdot 64^{\ell}\cdot
\max\left\{\left(\sqrt{\frac{6\ell}{k}}\right)^\ell, k\cdot
  \left(\frac{2\ell}{k}\right)^{\ell}, k\cdot
  \left(2c^2\ell^2\right)^\ell\right\} $$
as long as $h,\sigma$ are $\ell$-wise independent.  One can then set
$\ell = O(\log(1/\delta))$ and $c =
O((\sqrt{\eps}/\log(1/\delta))\cdot \eps^{2/\log(1/\delta)}) =
O(\sqrt{\eps^{1+o(1)}}/\log(1/\delta))$ to make the above probability
at most $\delta$.}
\end{remark}

\section{A high probability bound on $\|T\|_F$}\SectionName{frobenius}
In this section we prove \Lemma{frobenius}.

\vspace{.1in}

\begin{proofof}{\Lemma{frobenius}}
Recall that for $s,t\in[d]$, $\eta_{s,t}$ is the random variable
indicating that $s\neq t$ and $h(s) = h(t)$. Then, \Equation{Zdef}
implies that $\|T\|_F^2 = 2\sum_{s<t}x_s^2x_t^2\eta_{s,t}$. Note
$\|T\|_F^2$ is a random variable depending only on $h$. The plan of
our proof is to directly bound the $\ell$th moment of $\|T\|_F^2$ for
some large $\ell$ (specifically, $\ell = \Theta(\log(1/\delta))$),
then conclude by applying Markov's inequality to the random variable
$\|T\|_F^{2\ell}$. We bound the $\ell$th moment of $\|T\|_F^2$ via some
combinatorics.

We now give the details of our proof. Consider the expansion
$(\|T\|_F^2)^\ell$.  We have
\begin{equation}
 (\|T\|_F^2)^\ell =
2^\ell\cdot\sum_{\substack{s_1,\ldots,s_\ell\\t_1,\ldots,t_\ell\\\forall
    i\in[\ell] s_i < t_i}} \prod_{i=1}^\ell
x_{s_i}^2x_{t_i}^2\eta_{s_i,t_i}\EquationName{expand}
\end{equation}
Let $\mathcal{G}_\ell$ be the set of all isomorphism classes of
graphs (possibly containing multi-edges) with between $2$
and $2\ell$ unlabeled vertices, minimum degree at least $1$,
and exactly $\ell$ edges with distinct labels in $[\ell]$.
We now define a map
$f:\{\binom{[d]}{2}^\ell\}\rightarrow\mathcal{G}_\ell$ where
the notation $\binom{U}{r}$ denotes subsets of $U$ of size $r$;
i.e. $f$ maps the monomials in \Equation{expand} to
elements of $\mathcal{G}_\ell$.
Focus on one monomial in \Equation{expand} and let $S =
\{s_1,\ldots,s_\ell,t_1,\ldots,t_\ell\}$. We map the monomial
to an $|S|$-vertex element of $\mathcal{G}_\ell$ as follows:
associate each $u\in S$ with a vertex, and for each
$s_i,t_i$, draw an edge from the vertices associated with $s_i,t_i$
using edge label $i$.

We now analyze the expectation of the summation in \Equation{expand}
by grouping
monomials which map to the same elements of $\mathcal{G}_\ell$ under
$f$.

\begin{equation}
\E_h\left[(\|T\|_F^2)^\ell\right] =
2^\ell\cdot\sum_{G\in\mathcal{G}_\ell}
\sum_{\substack{\{(s_i,t_i)\}\in\binom{[d]}{2}^\ell
\\f(\{(s_i,t_i)\}) = G}}
\left(\prod_{i=1}^\ell
x_{s_i}^2x_{t_i}^2\right)\cdot
\E_h\left[\prod_{i=1}^\ell\eta_{s_i,t_i}\right] .\EquationName{graph}
\end{equation}

Observe that
$\prod_{i=1}^\ell \eta_{s_i,t_i}$ is determined by $h(s_i),h(t_i)$ for
each $i\in[\ell]$,
and hence its expectation is determined by $2\ell$-wise independence
of $h$.
Note that this product is $1$ if $s_i$ and $t_i$ hash to the same element
for each $i$ and is $0$ otherwise.  Each $s_i, t_i$ pair hash
to the same
element if and only if for each connected component of $G$, all
elements of $S = \{s_1,\ldots,s_\ell,t_1,\ldots,t_\ell\}$
corresponding to vertices in that component hash to the same
value.
For the $v_G$ elements we are concerned with, where $v_G = |S|$ is the
number of vertices in $G$, we can
choose one element of $[k]$ for each connected component.  Hence the
number
of possible values of $h$ on $S$ that cause $\prod_{i=1}^\ell
\eta_{s_i,t_i}$ to be $1$ is
$k^{m_G}$, where $G$ has $m_G$ connected components.
 Each possibility happens with probability $k^{-v_G}$.  Hence
$\E_h[\prod_{i-1}^\ell \eta_{s_i,t_i}] = k^{m_G - v_G}$.

Also,
consider the term $\prod_{i=1}^\ell x_{s_i}^2x_{t_i}^2 = \prod_{i=1}^{v_G}
x_{r_i}^{2\cdot \ell_i}$, where $S = \{r_i\}_{i=1}^{v_G}$, each $\ell_i$
is at least $1$, and $\sum_i \ell_i = 2\ell$ ($\ell_i$ is just the
degree of the vertex associated with $r_i$ in
$G$). Then,
$$\prod_{i=1}^{v_G} x_{r_i}^{2\cdot \ell_i} = \left(\prod_{i=1}^{v_G}
x_{r_i}^{2\cdot (\ell_i - 1)}\right) \cdot \left(\prod_{i=1}^{v_G}
x_{r_i}^2\right) \le \left(\prod_{i=1}^{v_G}
x_{r_i}^{2\cdot (\ell_i - 1)}\right) \cdot \left(\prod_{i=1}^{v_G}
x_{r_i}^2\right) \le c^{2(2\ell - v_G)}\cdot \left(\prod_{i=1}^{v_G}
x_{r_i}^2\right) .$$
Note then that the monomials $(\prod_{i=1}^{v_G}x_{r_i}^2)$ that arise
from the summation over $\{(s_i,t_i)\}\in\binom{d}{2}^\ell$ with
$f(\{(s_i,t_i)\}) = G$
in \Equation{graph}
are a subset of those monomials which appear in the expansion of
$(\sum_{i=1}^d x_i^2)^{v_G} = 1$.  Thus,
plugging back into \Equation{graph},
\begin{equation}
 \E_h\left[(\|T\|_F^2)^\ell\right] \le 2^\ell\cdot \sum_{G\in\mathcal{G}_\ell}
\frac{c^{2(2\ell - v_G)}}{k^{v_G - m_G}} .\EquationName{magic}
\end{equation}

Note the value $\ell$ in the $c^{2(2\ell - v_G)}$ term just arose as
$e_G$, the number of edges in $G$. We bound the above summation by
considering all ways to form an element of $\mathcal{G}_\ell$ by adding
one
edge at a time, starting from the empty graph $G_0$ with zero
vertices and edges. In fact we will overcount some
$G\in\mathcal{G}_\ell$, but this is acceptable since we only want an
upper bound on \Equation{magic}.

Define $F(G) = c^{2(2e_G - v_G)}/k^{v_G - m_G}$.
Initially
we have $F(G_0) = 1$. We will add $\ell$ edges in order by label, from
label $1$
to $\ell$. For the $i$th edge we have three options to form $G_i$
from $G_{i-1}$:
(a) we can add the edge between two existing vertices in $G_{i-1}$,
(b) we can add two new vertices to $G_{i-1}$ and place the edge
between them, or (c) we can create one new vertex and connect it to an
already-existing vertex of $G_{i-1}$. For each of these three options,
we will argue
that $n_i\cdot F(G_i)/F(G_{i-1})\le 1/k$,
where $n_i$ is the
number of ways to perform the operation we chose at step $i$.
This implies that the
right hand side of \Equation{magic} is at most $(6/k)^{\ell}$ since at each
step of forming an element of $\mathcal{G}_\ell$ we have three
options for how to form $G_i$ from $G_{i-1}$.

Let $e$ be the number of edges, $v$ the number of vertices, and $m$
the number of connected components for some $G_{i-1}$. In option (a),
$v$ remains constant, $e$ increases by $1$, and $m$ either remains
constant or decreases by $1$. In any case, $F(G_i)/F(G_{i-1}) \le
c^4$, and $n_i < 2\ell^2$; the latter is because we have
$\binom{v}{2} < 2\ell^2$
choices of vertices to connect. In option (b), $n_i = 1$, $v$
increases by $2$,
$e$ increases by $1$, and
$m$ increases by $1$, implying $n_i\cdot F(G_i)/F(G_{i-1}) =
1/k$. Finally, in
option (c), $n_i = v\le 2\ell$, $v$ increases by $1$, $e$ increases by
$1$, and $m$
remains constant, implying $n_i\cdot F(G_i)/F(G_{i-1}) \le
2\ell c^2/k$.  Thus,
regardless of which of the three options we choose,
$n_i\cdot F(G_i)/F(G_{i-1}) \le \max\{2\ell^2c^4, 1/k, 2\ell c^2/k\}$,
which is $1/k$ for $\ell = O(\log(1/\delta))$.

As discussed
above, when combined with \Equation{magic} this gives
\begin{equation}\E_h[(\|T\|_F^2)^\ell] \le (6/k)^\ell
  .\EquationName{frobenius}\end{equation}
 Then, by Markov's inequality on
the random variable $(\|T\|_F^2)^\ell$ for $\ell\ge 2$ and even, and
assuming $2\ell \le r_h$,
$$ \Pr_h[\|T\|_F^2 > 7/k] < (k/7)^{\ell}\cdot
\E_h[(\|T\|_F^2)^\ell] < (6/7)^\ell ,$$
which is at most $\delta$ for $\ell = \Theta(\log(1/\delta))$.
\end{proofof}


\section{A high probability bound on $\|T\|_2$}\SectionName{operator}
In this section we prove \Lemma{operator}.  For each
$j\in[k]$ we use
$\alpha_j$ to
denote $\sum_{\substack{i\in [d]\\h(i) =
    j}} x_i^2$.

\begin{lemma}\LemmaName{eigenbound}
$\|T\|_2 \le \max\{c^2,\max_{j\in[k]} \alpha_j\}$.
\end{lemma}
\begin{proof}
Define the diagonal matrix $R$ with $R_{i,i} = x_i^2$, and put $S = T
+ R$. 
For each $j\in[k]$, consider the vector $v_j$ whose support is
$h^{-1}(j)$, with $(v_j)_i = x_i$ for each $i$ in its support. 
Then $S = \sum_{j=1}^k v_j\cdot v_j^T$. Thus $\mathrm{rank}(S)$ is
equal to the number of non-zero $v_j$, since they are clearly linearly
independent (they have disjoint support and are thus orthogonal) and
span the image of $S$. Furthermore, these non-zero $v_j$ are
eigenvectors of $S$ since $Sv_j= \alpha_j v_j$, and are the
only eigenvectors of $S$ with non-zero eigenvalue since if $u$ is
perpendicular to all such $v_j$ then $Au = 0$.

Now, $\|T\|_2 = \sup_{\|x\|_2 =
  1}|x^TTx| = \sup_{\|x\|_2 =
  1}|x^TSx - x^TRx|$. Since
$S,R$ are both positive semidefinite, we then have $\|T\|_2 \le
\max\{\|S\|_2, \|R\|_2\}$. $\|R\|_2$ is clearly
$\|x\|_\infty^2 \le c^2$, and we saw above that $\|S\|_2 =
\max_{j\in[k]} \alpha_j$.
\end{proof}

\vspace{.1in}

\begin{proofof}{\Lemma{operator}}
Fix some $j\in[k]$. Define $X_i =
x_i^2\delta_{i,j}$ so that $\alpha_j = \sum_{i=1}^d X_i$.  Then
\begin{equation}\EquationName{expand-alpha}
 \E_h[\alpha_j^\ell] = \sum_{1\le s_1,\ldots,s_\ell\le d}
\E_h\left[\prod_{i=1}^\ell X_{s_i}\right]
\end{equation}
Let $V_\ell$ be the set of
length-$\ell$ vectors $v$ with
non-negative integer entries such that if $r>0$ appears as an entry of
$v$, then at least
one appearance of $r-1$ is in $v$ at an earlier index.  
Define the map $f:[d]^\ell\rightarrow V_{\ell}$ as follows: a vector
$w\in[d]^\ell$ maps to the vector where for each $i\in[\ell]$, if
$w_i$ is the $r$th distinct value ($0$-based indexing) to appear in
$w$ then we replace $w_i$ with $r$.  For example, $f((14,1,4,14)) =
(0, 1, 2, 0)$.  We group the monomials in \Equation{expand-alpha} by
equal images under $f$.  That is,
\begin{align*}
 \E_h[\alpha_j^\ell] & = \sum_{v\in V_{\ell}}\sum_{\substack{1\le
    s_1,\ldots,s_\ell \le d\\f((s_1,\ldots,s_\ell)) = v}}
\E_h\left[\prod_{i=1}^\ell X_{s_i}\right]
 = \sum_{v\in
  V_{\ell}}\sum_{\substack{1\le
    s_1,\ldots,s_\ell \le d\\f((s_1,\ldots,s_\ell)) = v}}
\left(\prod_{i=1}^\ell x_{s_i}^2\right)
\cdot\E_h\left[\prod_{i=1}^\ell\delta_{s_i,j}\right] \\
&{}= \sum_{v\in
  V_{\ell}}\sum_{\substack{1\le
    s_1,\ldots,s_\ell \le d\\f((s_1,\ldots,s_\ell)) = v}}
\left(\prod_{i=1}^\ell x_{s_i}^2\right)k^{-m_v} \le \sum_{v\in
  V_{\ell}}\sum_{\substack{1\le
    s_1,\ldots,s_\ell \le d\\f((s_1,\ldots,s_\ell)) = v}}
\frac{c^{2(\ell - m_v)}}{k^{m_v}}
\end{align*}
where the penultimate equality holds if $\ell \le r_h$ and $m_v$ is
the number of distinct values amongst the entries of $v$. 
The final equality holds since, pulling out a $c^2$ term for each
multiple occurrence of any $s_i$, for a fixed $v$ these terms
all show up in the expansion of $(\|x\|_2^2)^{m_v}$.

Now
we bound the double summation above.  Begin with the empty sequence
$v_0 = ()$
(in $V_0$).  We will arrive at some $v_\ell\in V_\ell$ by appending an
entry one at a time.  In transitioning from $v_{i-1}$ to $v_{i}$ we can
either (i) repeat an entry that already appeared in $v_{i-1}$, or
(ii) add a new entry (whose identity is unique: it must be the next
largest integer which has not appeared in $v_{i-1}$).  For (i) there
are $m_{v_{i-1}} \le \ell$ ways to choose a pre-existing integer to
repeat, $i$ increases by $1$, and $m_{v_i}=
m_{v_{i-1}}$, and thus we gain a factor of $c^2\ell$.  For (ii), there
is one way to choose a new integer to appear, $i$ increases by $1$, and
$m_{v_i} = m_{v_{i-1}} + 1$, and thus we gain a factor of $1/k$.  Since
  at each step we have two options to choose from (either perform (i)
  or (ii)), 
$\E_h[\alpha_j^\ell] \le 2^\ell \cdot \max\{1/k, c^2\ell\}^\ell$.
We then have 
\begin{equation}
 \E_h[\|T\|_2^\ell] \le \E_h\left[\left(\max_{j\in [k]}
     \alpha_j\right)^\ell\right] =
\E_h\left[\max_{j\in[k]}
\alpha_j^\ell\right]
\le \sum_{j=1}^k\E_h[\alpha_j^\ell] \le k\cdot 2^\ell
\cdot \max\{1/k, c^2\ell\}^\ell \EquationName{operator}
\end{equation}
The lemma follows by a Markov bound with $\ell = O(\log(k/\delta))$,
i.e. $\Pr_h[\|T\|_2 >
\lambda] < \lambda^{-\ell}\cdot \E_h[\|T\|_2^\ell]$ and we set
$\lambda = \eps/(128\cdot \log(1/\delta))$.
\end{proofof}

\begin{remark}
\textup{One could integrate the Bernstein inequality tail bound to obtain
a moment bound which applies to $\E_h[\alpha_j^\ell]$ in the proof of
\Lemma{operator}. The conclusion would not improve.
We chose to give an elementary proof
to be self-contained.}
\end{remark}

\section*{Acknowledgments}
This work was done as the authors interned at Microsoft
Research New England in Summer 2010.  

\bibliographystyle{plain}

\bibliography{allpapers}

\begin{thebibliography}{10}

\bibitem{Achlioptas03}
Dimitris Achlioptas.
\newblock Database-friendly random projections: {Johnson-Lindenstrauss} with
  binary coins.
\newblock {\em J. Comput. Syst. Sci.}, 66(4):671--687, 2003.

\bibitem{AC06}
Nir Ailon and Bernard Chazelle.
\newblock Approximate nearest neighbors and the fast {Johnson-Lindenstrauss}
  transform.
\newblock In {\em Proceedings of the 38th ACM Symposium on Theory of Computing
  (STOC)}, pages 557--563, 2006.

\bibitem{AL09}
Nir Ailon and Edo Liberty.
\newblock Fast dimension reduction using {Rademacher} series on dual {BCH}
  codes.
\newblock {\em Discrete Comput. Geom.}, 42(4):615--630, 2009.

\bibitem{AL11}
Nir Ailon and Edo Liberty.
\newblock Almost optimal unrestricted fast {Johnson-Lindenstrauss} transform.
\newblock In {\em Proceedings of the 22nd Annual ACM-SIAM Symposium on Discrete
  Algorithms (SODA), to appear}, 2011.

\bibitem{AMS99}
Noga Alon, Yossi Matias, and Mario Szegedy.
\newblock {The Space Complexity of Approximating the Frequency Moments}.
\newblock {\em J. Comput. Syst. Sci.}, 58(1):137--147, 1999.

\bibitem{AlonSpencer00}
Noga Alon and Joel~H. Spencer.
\newblock {\em The Probabilistic Method}.
\newblock Wiley-Interscience, 2nd edition, 2000.

\bibitem{AV06}
Rosa~I. Arriaga and Santosh Vempala.
\newblock An algorithmic theory of learning: Robust concepts and random
  projection.
\newblock {\em Machine Learning}, 63(2):161--182, 2006.

\bibitem{CCF02}
Moses Charikar, Kevin Chen, and Martin Farach-Colton.
\newblock Finding frequent items in data streams.
\newblock In {\em Proceedings of the 29th International Colloquium on Automata,
  Languages and Programming (ICALP)}, pages 693--703, 2002.

\bibitem{CW09}
Kenneth~L. Clarkson and David~P. Woodruff.
\newblock Numerical linear algebra in the streaming model.
\newblock In {\em Proceedings of the 41st ACM Symposium on Theory of Computing
  (STOC)}, pages 205--214, 2009.

\bibitem{DKS10}
Anirban Dasgupta, Ravi Kumar, and Tam{\'a}s Sarl{\'o}s.
\newblock A sparse {Johnson-Lindenstrauss} transform.
\newblock In {\em Proceedings of the 42nd ACM Symposium on Theory of Computing
  (STOC)}, pages 341--350, 2010.

\bibitem{DG03}
Sanjoy Dasgupta and Anupam Gupta.
\newblock An elementary proof of a theorem of {Johnson} and {Lindenstrauss}.
\newblock {\em Random Struct. Algorithms}, 22(1):60--65, 2003.

\bibitem{DKN10}
Ilias Diakonikolas, Daniel~M. Kane, and Jelani Nelson.
\newblock Bounded independence fools degree-2 threshold functions.
\newblock In {\em Proceedings of the 51st Annual IEEE Symposium on Foundations
  of Computer Science (FOCS), to appear (see also CoRR abs/0911.3389)}, 2010.

\bibitem{EIO02}
Lars Engebretsen, Piotr Indyk, and Ryan O'Donnell.
\newblock Derandomized dimensionality reduction with applications.
\newblock In {\em Proceedings of the 13th Annual ACM-SIAM Symposium on Discrete
  Algorithms (SODA)}, pages 705--712, 2002.

\bibitem{FM88}
Peter Frankl and Hiroshi Maehara.
\newblock The {Johnson-Lindenstrauss} lemma and the sphericity of some graphs.
\newblock {\em J. Comb. Theory. Ser. B}, 44(3):355--362, 1988.

\bibitem{HW71}
David~Lee Hanson and Farroll~Tim Wright.
\newblock A bound on tail probabilities for quadratic forms in independent
  random variables.
\newblock {\em Ann. Math. Statist.}, 42(3):1079--1083, 1971.

\bibitem{HV10}
Aicke Hinrichs and Jan Vyb\'{i}ral.
\newblock {Johnson-Lindenstrauss} lemma for circulant matrices.
\newblock {\em arXiv}, abs/1001.4919, 2010.

\bibitem{Indyk06}
Piotr Indyk.
\newblock Stable distributions, pseudorandom generators, embeddings, and data
  stream computation.
\newblock {\em J. ACM}, 53(3):307--323, 2006.

\bibitem{IM98}
Piotr Indyk and Rajeev Motwani.
\newblock Approximate nearest neighbors: Towards removing the curse of
  dimensionality.
\newblock In {\em Proceedings of the 30th ACM Symposium on Theory of Computing
  (STOC)}, pages 604--613, 1998.

\bibitem{JW11}
T.~S. Jayram and David~P. Woodruff.
\newblock Optimal bounds for {Johnson-Lindenstrauss} transforms and streaming
  problems with low error.
\newblock In {\em Proceedings of the 22nd Annual ACM-SIAM Symposium on Discrete
  Algorithms (SODA), to appear}, 2011.

\bibitem{JL84}
William~B. Johnson and Joram Lindenstrauss.
\newblock Extensions of {Lipschitz} mappings into a {Hilbert} space.
\newblock {\em Contemporary Mathematics}, 26:189--206, 1984.

\bibitem{KRS09}
Zohar Karnin, Yuval Rabani, and Amir Shpilka.
\newblock Explicit dimension reduction and its applications.
\newblock {\em Electronic Colloquium on Computational Complexity (ECCC)},
  (121), 2009.

\bibitem{KW10}
Felix Krahmer and Rachel Ward.
\newblock New and improved {Johnson-Lindenstrauss} embeddings via the
  {Restricted Isometry Property}.
\newblock {\em arXiv}, abs/1009.0744, 2010.

\bibitem{LAS08}
Edo Liberty, Nir Ailon, and Amit Singer.
\newblock Dense fast random projections and {Lean Walsh} transforms.
\newblock In {\em Proceedings of the 12th International Workshop on
  Randomization and Computation (RANDOM)}, pages 512--522, 2008.

\bibitem{Matousek08}
Jir\'{\i} Matousek.
\newblock On variants of the {Johnson-Lindenstrauss} lemma.
\newblock {\em Random Struct. Algorithms}, 33(2):142--156, 2008.

\bibitem{Meka10}
Raghu Meka.
\newblock Almost optimal explicit {Johnson-Lindenstrauss} transformations.
\newblock {\em CoRR}, abs/1011.6397, 2010.

\bibitem{MZ10}
Raghu Meka and David Zuckerman.
\newblock Pseudorandom generators for polynomial threshold functions.
\newblock In {\em Proceedings of the 42nd Annual ACM Symposium on Theory of
  Computing (STOC), to appear (see also CoRR abs/0910.4122)}, 2010.

\bibitem{Muthu}
S.~Muthukrishnan.
\newblock {Data Streams: Algorithms and Applications}.
\newblock {\em Foundations and Trends in Theoretical Computer Science},
  1(2):117--236, 2005.

\bibitem{Nisan92}
Noam Nisan.
\newblock Pseudorandom generators for space-bounded computation.
\newblock {\em Combinatorica}, 12(4):449--461, 1992.

\bibitem{Sivakumar02}
D.~Sivakumar.
\newblock Algorithmic derandomization via complexity theory.
\newblock In {\em Proceedings of the 34th Annual ACM Symposium on Theory of
  Computing (STOC)}, pages 619--626, 2002.

\bibitem{ThorupZhang04}
Mikkel Thorup and Yin Zhang.
\newblock Tabulation based 4-universal hashing with applications to second
  moment estimation.
\newblock In {\em Proceedings of the 15th Annual ACM-SIAM Symposium on Discrete
  Algorithms (SODA)}, pages 615--624, 2004.

\bibitem{GG99}
Joachim von~zur Gathen and J\"{u}rgen Gerhard.
\newblock {\em Modern Computer Algebra}.
\newblock Cambridge University Press, 1999.

\bibitem{Vybiral10}
Jan Vyb\'{i}ral.
\newblock A variant of the {Johnson-Lindenstrauss} lemma for circulant
  matrices.
\newblock {\em arXiv}, abs/1002.2847, 2010.

\end{thebibliography}

\appendix

\section*{Appendix}\SectionName{appendix}

\section{Optimality of \Lemma{jl-lemma}}\SectionName{jl-optimal}
Jayram and Woodruff gave a proof that the $k =
\Omega(\eps^{-2}\log(1/\delta))$ in \Lemma{jl-lemma} is optimal
\cite{JW11}.  Their proof went through communication and information
complexity.
We here give another proof of this fact, via some linear
algebra and direct calculations.

Note that for a distribution $\mathcal{D}$, if $\Pr_{A\sim
  \mathcal{D}}[|\|Ax\|_2^2 - 1|] < \delta]$ for any $x\in S^{d-1}$,
then it must be the case that $\Pr_{A\sim \mathcal{D}}[\Pr_{x\in
  S^{d-1}}[|\|Ax\|_2^2 - 1|]] < \delta$.  The following theorem shows
that no $A\in\R^{k\times d}$ can have $\Pr_{x\in
  S^{d-1}}[|\|Ax\|_2^2 - 1|] < \delta$ unless $k$ is at least as large
as in the statement of \Lemma{jl-lemma}.

\begin{theorem}
If $A:\R^d\rightarrow\R^k$ is a linear transformation with $d > 2k$
and $\epsilon>0$ sufficiently small, then for $x$ a randomly chosen
vector in $S^{d-1}$, $\Pr[|\|Ax\|_2^2-1| > \epsilon] \ge
\exp(-O(k\epsilon^2+1))$.
\end{theorem}
\begin{proof}
First we note that we can assume that $A$ is surjective since if it is
not, we may replace $\R^k$ by the image of $A$.  Let $V=\ker(A)$ and
let $W=V^\perp$.  Then $\dim(W)=k$, $\dim(V)=d-k$.  Now, any $x\in
\R^d$ can be written uniquely as $x_V + x_W$ where $x_V$ and $x_W$ are
the components in $V$ and $W$ respectively.  We may then write $x_V =
r_V \Omega_V$, $x_W = r_W \Omega_W$, where $r_V,r_W$ are positive real
numbers and $\Omega_V$ and $\Omega_W$ are unit vectors in $V$ and $W$
respectively.  Let $s_V=r_V^2$ and $s_W=r_W^2$.  We may now
parameterize the unit sphere by $(s_V,\Omega_V,s_W,\Omega_W)\in
[0,1]\times S^{d-k-1}\times [0,1]\times S^{k-1}$, so that $s_V+s_W=1.$
It is clear that the uniform measure on the sphere is given in these
coordinates by $f(s_W)ds_Wd{\Omega_V}d{\Omega_W}$ for some function
$f$.  To compute $f$ we note that $f(s_W)$ should be proportional to
the limit as $\delta_1,\delta_2\rightarrow 0^+$ of
$(\delta_1\delta_2)^{-1}$ times the volume of points $x$ so that
$\|x\|_2^2 \in [1,1+\delta_1]$ and $\|x_W\|_2^2\in [s_W,s_W+\delta_2]$.
Equivalently, $\|x_W\|_2^2\in [s_W,s_W+\delta_2]$, and $\|x_V\|_2^2 \in
[1-\|x_W\|_2^2,1-\|x_W\|_2^2+\delta_1]$.  For fixed $x_W$, the latter volume
is within $O(\delta_1\delta_2)$ of the volume of $x_V$ so that
$\|x_V\|_2^2 \in [s_V,s_V+\delta_1]$.  Now the measure on $V$ is
$r_V^{d-k-1}dr_Vd{\Omega_V}$.  Therefore it also is
$\frac{1}{2}s_V^{(d-k-2)/2}ds_Vd\Omega_V$.  Therefore this volume over
$V$ is proportional to
$s_V^{(d-k-2)/2}(\delta_1+O(\delta_1\delta_2+\delta_1^2))$.  Similarly
the volume of $x_W$ so that $\|x_W\|_2^2\in [s_W,s_W+\delta_2]$ is
proportional to $s_W^{(k-2)/2}(\delta_2+O(\delta_2^2))$.  Hence $f$ is
proportional to $s_V^{(d-k-2)/2}s_W^{(k-2)/2}$.

We are now prepared to prove the theorem.  The basic idea is to
first condition on $\Omega_V,\Omega_W$.  We let $C=\|A\Omega_W\|_2^2$.
Then if $x$ is parameterized by $(s_V,\Omega_V,s_W,\Omega_W)$,
$\|Ax\|_2^2 = Cs_W$.  Choosing $x$ randomly, we know that $s=s_W$
satisfies the distribution
$\frac{s^{(k-2)/2}(1-s)^{(d-k-2)/2}}{\beta((k-2)/2,(d-k-2)/2)}ds=f(s)ds$
on $[0,1].$  We need to show that for any $c=\frac{1}{C}$, the
probability that $s$ is not in $[(1-\epsilon)c,(1+\epsilon)c]$ is
$\exp(-O(\epsilon^2 k))$.  Note that $f(s)$ attains its maximum value
at $s_0 = \frac{k-2}{d-4} < \frac{1}{2}$.  Notice that
$\log(f(s_0(1+x)))$ is some constant plus $\frac{k-2}{2}\log(s_0(1+x))
+\frac{d-k-2}{2}\log(1-s_0-xs_0)$.  If $\|x\|_2<1/2$, then this is some
constant plus $-O(kx^2)$. So for such $x$, $f(s_0(1+x)) =
f(s_0)\exp(-O(kx^2))$. Furthermore, for all $x$, $f(s_0(1+x)) =
f(s_0)\exp(-\Omega(kx^2))$. This says that $f$ is bounded above by a
normal distribution and checking the normalization we find that
$f(s_0) = \Omega(s_0^{-1}k^{1/2})$.
 
We now show that both $\Pr(s<(1-\epsilon)s_0)$ and
$\Pr(s>(1+\epsilon)s_0)$ are reasonably large.  We can lower bound
either as
\begin{align*}
s_0\int_{\epsilon}^{1/2}f(s_0)\exp(-O(kx^2))dx & \geq
\Omega(k^{1/2})\int_{\epsilon}^{\epsilon+k^{-1/2}}\exp(-O(kx^2))dx\\
& \geq \Omega(\exp(-O(k(\epsilon+k^{-1/2})^2)))\\
& \geq \exp(-O(k\epsilon^2+1)).
\end{align*}
Hence since one of these intervals is disjoint from
$[(1-\epsilon)c,(1+\epsilon)c]$, the probability that $s$ is not in
$[(1-\epsilon)c,(1+\epsilon)c]$ is at least $\exp(-O(k\epsilon^2
+1)).$
\end{proof}

\section{A JL Lemma derandomization}\SectionName{derandomize}
We give an explicit JL family with seed length
$O(\log d + \log(1/\eps)\log(1/\delta) +
\log(1/\delta)\log\log(1/\delta))$.  This seed length is always at
least at least as good as the $O(\log(1/\delta)\log d)$ seed length
coming from $k$-wise independence, but can be much better for some
settings of parameters.

The idea is simply to graduately reduce the dimension.   Consider
values $\eps',\delta'>0$ which we will pick later.
Define $t_j = \delta'^{-1/2^j}$.
We embed
$\R^d$ into $\R^{k_1}$ for $k_1 = \eps'^{-2}t_1$.  We then embed
$\R^{k_{j-1}}$ into $\R^{k_{j}}$ for $k_j = \eps'^{-2}t_j$ until the
point $j = j^* = O(\log(1/\delta')/\log\log(1/\delta'))$ where $t_{j^*}
= O(\log^3(1/\delta'))$.  We then embed $\R^{k_{j^*}}$ into $\R^k$ for
$k = O(\eps'^{-2}\log(1/\delta'))$.

The embeddings into each $k_j$ are performed by picking a Bernoulli
matrix with $r_j$-wise independent entries, as in \Theorem{warmup}.
To achieve error probability $\delta'$ of having distortion larger than
$(1+\eps')$ in the $j$th step, \Equation{warmup-eq} in the proof of
\Theorem{warmup} tells us we need $r_j =
O(\log(1/\delta')/(t_j/\log(1/\delta')))$.  Thus, in the first embedding
into $\R^{k_1}$ we need $O(\log d)$ random bits.  Then in the future
embeddings except the last, we need $O(2^j\cdot (\log(1/\eps') +
2^{-j}\log(1/\delta')))$ random bits to embed into $\R^{k_j}$.  In the
final embedding we require $O(\log(1/\delta')\cdot (\log(1/\eps') +
\log\log(1/\delta'))$ random bits.  Thus, in total, we have used
$O(\log d + \log(1/\eps')\log(1/\delta') +
\log(1/\delta')\log\log(1/\delta'))$ bits of
seed to achieve error probability
$O(\delta'\cdot j^*)$ of distortion $(1+\eps')^{j^*}$.  The following theorem
thus follows by applying this argument with error probability $\delta'
= \delta/(j^*+1)$ and distortion parameter $\eps' = \Theta(\eps/j^*)$.

\begin{theorem}
For any $0<\eps,\delta<1/2$ there exists an explicit JL family with
seed length $s =
O(\log d +
\log(1/\eps)\log(1/\delta) + \log(1/\delta)\log\log(1/\delta))$. Given
a seed and a vector $x\in\R^d$, the embedding can be performed in
polynomial time.
\end{theorem}

\begin{remark}
One may worry that along the way we are embedding into potentially
very large dimension (e.g. $1/\delta$ may be $2^{\Omega(d)}$), so that
our overall running time could be exponentially large.  However, we
can simply start the above iterative embedding at the level $j$ where
$\eps^{-2}t_j < d$.
\end{remark}

\end{document}